\documentclass[11pt]{article}
\usepackage{amssymb}

\usepackage{graphicx}
\usepackage{amsmath}
\usepackage{mathpazo}
\usepackage{euscript}
\usepackage[pdftex]{hyperref}
\numberwithin{equation}{section}
\makeatletter\@addtoreset{equation}{section}

\newtheorem{theorem}{Theorem}[section]

\newtheorem{definition}[theorem]{Definition}

\newtheorem{proposition}[theorem]{Proposition}
\newtheorem{remark}[theorem]{Remark}

\newenvironment{proof}[1][Proof]{\textbf{#1.} }{\ \rule{0.5em}{0.5em}}

\begin{document}

\title{\textsf{Phase\ coherent states with circular Jacobi polynomials for the
pseudoharmonic oscillator}}
\author{Zouha\"{i}r MOUAYN}
\date{{\small Sultan Moulay Slimane University, Faculty of \ Sciences and Technics
(M'Ghila)  PO.Box 523, B\'{e}ni Mellal, Morocco }\\
{\small mouayn@gmail.com}}
\maketitle

\begin{abstract}
We construct a class of generalized phase coherent states $\mid e^{i\theta
};\gamma ,\alpha ,\varepsilon >$ indexed by points $e^{i\theta }$ of the
unit circle and depending on three positive parameters $\gamma $, $\alpha $
and $\varepsilon $ by replacing the labeling coefficient $z^{n}/\sqrt{n!}$
of the canonical coherent states by circular Jacobi polynomials $%
g_{n}^{\gamma }\left( e^{i\theta }\right) $ with parameter $\gamma \geq 0.$
The special case of $\gamma =0$ corresponds to well known phase coherent states.
The constructed states are superposition of eigenstates of a
one-parameter pseudoharmonic oscillator depending on $\alpha $ and solve the
identity of the state Hilbert space at the limit $\varepsilon \rightarrow
0^{+}.$ Closed form for their wavefunctions are obtained in the case $\alpha
=\gamma +1$ and their associated coherent states transform is defined.
\end{abstract}

\section{Introduction}

The first phase states which appeared in quantum mechanics were the London
phase states \cite{1}:
\begin{equation}
\mid e^{i\theta }>=\sum\limits_{n=0}^{+\infty }e^{in\theta }\mid n>
\label{1.1}
\end{equation}
which actually possess infinite energy and can be approached by physical
states in different ways. Consequently, different phase-enhanced states have
been proposed and studied in the literature; see \cite{2,3} and references therein.

The phase analogue of coherent states have been introduced as eigenstates of
the Susskind-Glogower operator \cite{4}:
\begin{equation}
\mathcal{E}^{SG}:=\sum\limits_{n=0}^{+\infty }\mid n><n+\mid .  \label{1.2}
\end{equation}
There are the so-called phase coherent states (PCS), whose expression in the
Fock basis is given by
\begin{equation}
\mid \frak{\rho }e^{i\theta }>=\left( 1-\frak{\rho }^{2}\right) ^{\frac{1}{2}%
}\sum\limits_{n=0}^{+\infty }\frak{\rho }^{n}e^{in\theta }\mid n>; \quad 0\leq \rho
<1  \label{1.3}
\end{equation}
and represent realistic states of radiation, namely they possess finite
energy and can be synthesized by suitable nonlinear process \cite{5}.
 They can also be viewed as a special case of negative binomial states
\cite{6} or Perelomov's $su( 1,1)$ coherent states via
its Holstein-Primakoff realization with $\frac{1}{2}$ as Bargmann index \cite{7}.

In this paper, we construct a class of generalized phase coherent states
(GPCS) labeled by points $e^{i\theta }$ of the unit circle $S^{1}$ and
depending on three positive parameters: $\gamma ,\alpha $ and $\varepsilon .$
These states belong to the state Hilbert space $L^{2}(\mathbb{R}_{+},dx)$ of
the Hamiltonian with pseudoharmonic oscillator potential (PHO) given by (\cite{8}):
\begin{equation}
\Delta _{a}:=-\frac{d^{2}}{dx^{2}}+x^{2}+\frac{a}{x^{2}},  \label{1.4}
\end{equation}
where $a>0$ is such that $1+\frac{1}{2}\sqrt{1+4a}=\alpha .$

We precisely adopt a formalism of canonical coherent states when written as
superpositions of the harmonic oscillator number states. That is, we present
a GPCS as a superposition of eigenstates of the Hamiltonian in \eqref{1.4}. In this superposition, the role of coefficients $z^{n}/\sqrt{n!%
}$ is played by the circular Jacobi polynomials (\cite[p.230]{9}, \cite{10,11}) given, up to a normalization, by:
\begin{equation}
g_{n}^{\gamma }\left( e^{i\theta }\right) :=\left( n!\right) ^{-1}\left(
\gamma +1\right) _{n\text{ }2}F_{1}(-n,\frac{\gamma }{2}+1,\gamma
+1,1-e^{i\theta })  \label{1.5}
\end{equation}
where $_{2}F_{1}$ is the Gauss hypergeometric function and $\left( \cdot
\right) _{n}$ denotes Pochhammer symbol. These orthogonal polynomials arise
in a class of random matrix ensembles, where the parameter $\gamma $ is
related to the charge of an impurity fixed at $z=1$ in a system of unit
charges located on $S^{1}$ at the complex values given by eigenvalues of a
member of this matrix ensemble \cite{9}, \cite{12}.

The class of GPCS we are introducing contains the form of the well known PCS
in \eqref{1.3} as the special case $\gamma =0.$ The identity of the
state Hilbert space $L^{2}\left( \mathbb{R}_{+},dx\right) $ carrying the
GCPS is solved at the limit $\varepsilon \rightarrow 0^{+}$ by a similar way
in a previous work \cite{13}. Furthermore, if we link $\gamma $
with the parameter $\alpha $ controlling the singular part of PHO potential
in \eqref{1.4} by $\alpha =\gamma +1,$ then we can establish a
closed form for the constructed states. In this case, we propose a suitable
definition for the associated coherent states transform and we check that it
maps the eigenstates of the Hamiltonian $\Delta _{a},$ which are defined $\mathbb{R}_{+}$ onto the normalized circular Jacobi polynomials defined on $S^{1}$.

The paper is organized as follows. In Section 2, we recall briefly some
needed spectral properties of the Hamiltonian with PHO potential. Section 3
is devoted to the coherent states formalism we will be using. This formalism
is applied in Section 4 so as to construct a class of phase coherent states
in the state Hilbert space of the Hamiltonian. In Section 5 we give a closed
form for these states an we discuss their associated coherent states
transform.

\section{The pseudoharmonic oscillator}

The PHO potential was pointed out in \cite{8} and \cite{14} and studied by many authors (see \cite{15} and references
therein). It can be used to calculate the vibrational energies of a diatomic
molecules, with the form
\begin{equation}
V_{\varrho ,\kappa _{0}}\left( x\right) :=\varrho \left( \frac{x}{\kappa _{0}%
}-\frac{\kappa _{0}}{x}\right) ^{2} , \label{2.1}
\end{equation}
where $\kappa _{0}$ $>0$ denotes the equilibrium bond length which is the
distance between the diatomic nuclei, and \ $\varrho >0$ with $F=\varrho
\kappa _{0}^{-2}$ represents a constant force. The associated stationary Schr\"{o}dinger equation reads
\begin{equation}
-\frac{d^{2}}{dx^{2}}\psi \left( x\right) +\varrho \left( \frac{x}{\kappa
_{0}}-\frac{\kappa _{0}}{x}\right) ^{2}\psi \left( x\right) =\lambda \psi
\left( x\right) ,  \label{2.2}
\end{equation}
with $\psi \left( 0\right) =0,$ namely $\psi $ satisfies the Dirichlet
boundary condition. It is an exactly solvable equation. Indeed, according to
\cite[p.11288]{16}) the energy spectrum is given by
\begin{equation}
\lambda _{n}^{\varrho ,\kappa _{0}}:=4\kappa _{0}^{-1}\sqrt{\varrho }\left(
n+\frac{1}{2}+\frac{1}{4}\left( \sqrt{1+4\varrho \kappa _{0}^{2}}-2\kappa
_{0}\sqrt{\varrho }\right) \right) ,n=0,1,2,\cdots  \label{2.3}
\end{equation}
whereas the wave functions of the exact solutions of \eqref{2.2} take the form
\begin{equation}
<x\mid n;\varrho ,\kappa _{0}>\propto x^{q}\exp \left( -\frac{\sqrt{\varrho }%
}{2\kappa _{0}}x^{2}\right) \text{ }_{1}\digamma _{1}\left( -n,q+1,\frac{%
\sqrt{\varrho }}{2\kappa _{0}}x^{2}\right) ,  \label{2.4}
\end{equation}
where $q=\frac{1}{2}\left( 1+\sqrt{1+4\varrho \kappa _{0}^{2}}\right) $ and $%
_{1}\digamma _{1}$ denotes the confluent hypergeometric function which can
also be expressed in terms of Laguerre polynomials as (\cite[p.240]{17}):
\begin{equation}
_{1}\digamma _{1}\left( -n,\nu ,u\right) =\frac{n!}{\left( \nu \right) _{n}}%
L_{n}^{\left( \nu -1\right) }\left( u\right) ,  \label{2.5}
\end{equation}
where the Pochhammer symbol may also be defined by the Euler gamma function as
\begin{equation}
\left( \nu \right) _{0}=1,\left( \nu \right) _{n}=\nu \left( \nu +1\right)
\cdots\left( \nu +n-1\right) =\frac{\Gamma \left( \nu +n\right) }{\Gamma \left(
\nu \right) }; \quad n=1,2, \cdots.  \label{2.6}
\end{equation}
To simplify the notation, we introduce the parameters $a:=\varrho \kappa
_{0}^{2}$ and \ we put $\ \kappa _{0}^{-1}\sqrt{\varrho }=1$, and thereby
the Hamiltonian in \eqref{2.2} takes the form given in \eqref{1.4} by the operator $\Delta _{a}$ which is also called isotonic
oscillator \cite{18} or Gol'dman-Krivchenkov Hamiltonian \cite{16}.
The spectrum of $\Delta _{a}$ in\ the Hilbert space $L^{2}\left(
\mathbb{R}_{+},dx\right) $\ reduces to its discrete part consisting of
eigenvalues of the form (\cite[pp.9-10]{19}):
\begin{equation}
\lambda _{n}^{\alpha }:=2\left( 2n+\alpha \right) ,\alpha =1+\frac{1}{2}%
\sqrt{1+4a}; \quad n=0,1,2,\cdots,  \label{2.7}
\end{equation}
and wavefunctions of the corresponding normalized eigenfunctions are given by
\begin{equation}
<x\mid n;\alpha >:=\left( \frac{2n!}{\Gamma \left( \alpha +n\right) }\right)
^{\frac{1}{2}}x^{\alpha -\frac{1}{2}}e^{-\frac{1}{2}x^{2}}L_{n}^{\left(
\alpha -1\right) }\left( x^{2}\right); \quad n=0,1,2, \cdots .  \label{2.8}
\end{equation}
The set of functions in \eqref{2.8} constitutes a complete
orthonormal basis for the Hilbert space $L^{2}(\mathbb{R}_{+},dx).$

\begin{remark} We should note that the eigenvalue problem for the PHO
can also be considered by using raising and lowering operators throughout a
factorization of the Hamiltonian $\Delta _{a}$ in \eqref{1.4} based
on the Lie algebra $su( 1,1)$ commutation relations \cite{15}. \end{remark}

\section{A coherent states formalism}

In general, coherent states are a specific overcomplete family of vectors in
the Hilbert space of the problem that describes the quantum phenomena and
solves the identity of this Hilbert space. These states have long been
known for the harmonic oscillator and their properties have frequently been
taken as models for defining this notion for other models \cite{20}. In this section, we adopt the generalization of canonical coherent states
as in \cite{13}, which extend a well known generalization \cite{18} by considering a kind of the identity resolution that we obtain
as a limit with respect to a certain parameter. Precisely, we propose the
following formalism.

\begin{definition}
Let $\mathcal{H}$ be a separable Hilbert space with an orthonormal basis $%
\left\{ \psi _{n}\right\} _{n=0}^{+\infty }.$ Let $\frak{D}$ $\subseteq
\mathbb{C}$ be an open subset of $\mathbb{C}$ and let $\Phi _{n}:\frak{%
D\rightarrow }\mathbb{C}; \quad n=0,1,2,\cdots,$ be a sequence of complex functions.
Define
\begin{equation}
\mid z,\varepsilon >:=\left( N_{\varepsilon }\left( z\right) \right) ^{-%
\frac{1}{2}}\sum\limits_{n=0}^{+\infty }\frac{\Phi _{n}\left( z\right) }{%
\sqrt{\sigma _{\varepsilon }\left( n\right) }}\mid \psi _{n}>; \quad z\in \frak{D,}%
\varepsilon >0,  \label{3.1}
\end{equation}
where $N_{\varepsilon }\left( z\right) $ is a normalization factor and $%
\sigma _{\varepsilon }\left( n\right)$; $n=0,1,2,\cdots,$ a sequence of positive
numbers depending on $\varepsilon >0$. The set of vectors $\left\{ \mid
z,\varepsilon >,z\in \frak{D}\right\} $ is said to form a set of generalized
coherent states if :\newline
$\left( i\right) $ for each fixed $\varepsilon >0$ and $z\in \frak{D,}$ the
state $\mid z,\varepsilon >$ is normalized, that is $<z,\varepsilon \mid
z,\varepsilon >_{\mathcal{H}}=1,$\newline
$\left( ii\right) $ the states $\left\{ \mid z,\varepsilon >,z\in \frak{D}%
\right\} $ satisfy the following resolution of the identity
\begin{equation}
\lim_{\varepsilon \rightarrow 0^{+}}\int\limits_{\frak{D}}\mid z,\varepsilon
><z,\varepsilon \mid d\mu _{\varepsilon }\left( z\right) =\mathbf{1}_{%
\mathcal{H}}  \label{3.2}
\end{equation}
where $d\mu _{\varepsilon }$ is an appropriately chosen measure and $\mathbf{%
1}_{\mathcal{H}}$ is the identity operator on the Hilbert space $\mathcal{H}%
. $
\end{definition}

We should precise that, in the above definition, the Dirac's\textit{\ bra-ket%
} notation $\mid z,\varepsilon ><z,\varepsilon \mid $ means the
rank-one-operator $\varphi \longmapsto <\varphi \mid z,\varepsilon >_{%
\mathcal{H}}\mid z,\varepsilon >,$ $\varphi \in \mathcal{H}.$ Also, the
limit in $\left( ii\right) $ is to be understood as follows. Define the
operator
\begin{equation}
\mathcal{O}_{\varepsilon }\left[ \varphi \right] \left( \cdot \right)
:=\left( \int\limits_{\frak{D}}\mid z,\varepsilon ><z,\varepsilon \mid d\mu
_{\varepsilon }\left( z\right) \right) \left[ \varphi \right] \left( \cdot
\right)  \label{3.3}
\end{equation}
then the above limit \eqref{3.2} means that $\mathcal{O}_{\varepsilon }\left[
\varphi \right] \left( \cdot \right) \rightarrow $ $\varphi \left( \cdot
\right) $ as $\varepsilon \rightarrow 0^{+},$ \textit{almost every where }%
with respect to $\left( \cdot \right) .$

\begin{remark} The formula \eqref{3.1}\ can be considered as
a generalization of the series expansion of the canonical coherent states
\begin{equation}
\mid z>:=\left( e^{\left| z\right| ^{2}}\right) ^{-\frac{1}{2}%
}\sum_{k=0}^{+\infty }\frac{z^{k}}{\sqrt{k!}}\mid k>; \quad z\in \mathbb{C}.
\label{3.4}
\end{equation}
with $\mid k>$; $k=0,1,2,\cdots,$\ being an orthonormal basis in $L^{2}\left(
\mathbb{R},d\xi \right) $ of eigenstates of the harmonic oscillator, which
is given by the wavefunctions are $<\xi \mid k>:=\left( \sqrt{\pi }%
2^{k}k!\right) ^{-\frac{1}{2}}e^{-\frac{1}{2}\xi ^{2}}H_{k}(\xi )$ where $%
H_{k}\left( .\right) $ denotes the $k$th Hermite polynomial \cite{17}.
\end{remark}

\section{Generalized phase coherent states}

We now will construct a set of normalized states labeled by points $%
e^{i\theta }$ of the unit circle $S^{1}=\left\{ \omega \in \mathbb{C},\left|
\omega \right| =1\right\} $ and depending on positive parameters $\gamma $, $%
\alpha $ and $\varepsilon >0.$ These states, denoted by $\mid e^{i\theta
};\varepsilon ,\gamma ,\alpha >$, will belong to $L^{2}\left( \mathbb{R}%
_{+},dx\right) $ the state Hilbert space of the Hamiltonian $\Delta _{a}$ in
\eqref{1.4} as mentioned in the introduction.

\begin{definition}
Define a set of states $\mid e^{i\theta };\varepsilon ,\gamma ,\alpha >$
labeled by points $e^{i\theta }\in S^{1},\theta \in \left[ 0,2\pi \right] $
and depending on the parameters $\gamma \geq 0$, $\alpha >\frac{3}{2}$ and $%
\varepsilon >0$ by
\begin{equation}
\mid e^{i\theta };\varepsilon ,\gamma ,\alpha >:=\left( \mathcal{N}_{\gamma
,\varepsilon }\left( \theta \right) \right) ^{-\frac{1}{2}%
}\sum\limits_{n=0}^{+\infty }\frac{g_{n}^{\gamma }\left( e^{i\theta }\right)
}{\sqrt{\sigma _{\gamma ,\varepsilon }\left( n\right) }}\mid n;\alpha >
\label{4.1}
\end{equation}
with the precisions:
\begin{enumerate}
  \item[$\bullet$] $\mathcal{N}_{\gamma ,\varepsilon }\left( \theta \right) $
is a normalization factor such that $<e^{i\theta },\varepsilon ,\gamma
,\alpha \mid e^{i\theta },\varepsilon ,\gamma ,\alpha >=1$
  \item[$\bullet$] $g_{m}^{\gamma }\left( e^{i\theta }\right) $ are the circular
Jacobi polynomials defined in \eqref{1.5}
  \item[$\bullet$] $\sigma _{\gamma
,\varepsilon }\left( n\right) ,$ $n=0,1,2,\cdots$ are a sequence of positive
numbers given by
\begin{equation}
\sigma _{\gamma ,\varepsilon }\left( n\right) :=\left( n!\right) ^{-1}\left(
\gamma +1\right) _{n}e^{n\varepsilon },  \label{4.2}
\end{equation}
 \item[$\bullet$] $ \mid n;\alpha >; \quad n=0,1,2,\cdots$ , is the orthonormal basis of $%
L^{2}\left( \mathbb{R}_{+},dx\right) $ given in \eqref{2.6}.
\end{enumerate}
\end{definition}

We shall give the main properties on these states in two propositions.

\begin{proposition}
Let $\gamma \geq 0$ and $\varepsilon >0$ be fixed parameters. Then, the
normalization factor in \eqref{4.1} has the expression
\begin{equation}
\mathcal{N}_{\gamma ,\varepsilon }\left( \theta \right)
 =\frac{\left(1-e^{-\varepsilon }\right) }{\left| 1-e^{-\varepsilon +i\theta }\right|^{2+\gamma }}\; 
{_2F_1}\left(\frac{\gamma }{2}+1,\frac{\gamma }{2}+1,\gamma +1;\frac{e^{-\varepsilon }\left| 1-e^{i\theta }\right| ^{2}}{\left|
1-e^{-\varepsilon +i\theta }\right| ^{2}}\right)  \label{4.3}
\end{equation}
for every $\theta \in \left[ 0,2\pi \right] $.
\end{proposition}

\begin{proof}
To calculate this factor, we start by writing the condition
\begin{equation}
1=<e^{i\theta },\varepsilon ,\gamma ,\alpha \mid e^{i\theta },\varepsilon
,\gamma ,\alpha >.  \label{4.4}
\end{equation}
Eq.\eqref{4.4} is equivalent to
\begin{equation}
\left( \mathcal{N}_{\gamma ,\varepsilon }\left( \theta \right) \right)
^{-1}\sum\limits_{n=0}^{+\infty }\frac{1}{\sigma _{\gamma ,\varepsilon
}\left( n\right) }g_{n}^{\gamma }\left( e^{i\theta }\right) \overline{%
g_{n}^{\gamma }\left( e^{i\theta }\right) }=1.  \label{4.5}
\end{equation}
Inserting the expression \eqref{4.2} into \eqref{4.5},
we obtain that
\begin{equation}
\mathcal{N}_{\gamma ,\varepsilon }\left( \theta \right)
=\sum\limits_{n=0}^{+\infty }\frac{n!e^{-n\varepsilon }}{\left( \gamma
+1\right) _{n}}g_{n}^{\gamma }\left( e^{i\theta }\right) \overline{%
g_{n}^{\gamma }\left( e^{i\theta }\right) }.  \label{4.6}
\end{equation}
Explicitly, the sum in \eqref{4.6} reads
\begin{equation}
\sum\limits_{n=0}^{+\infty }\frac{\left( \gamma +1\right) _{n}}{%
n!e^{n\varepsilon }}\; {_2F_1}\left(-n,\frac{\gamma }{2}+1,\gamma
+1,1-e^{i\theta }\right) \; {_2F_1}\left(-n,\frac{\gamma }{2}+1,\gamma +1,1-e^{-i\theta
}\right).  \label{4.7}
\end{equation}
Making use of the formula (\cite[p.85]{21}):
\begin{align}
\sum\limits_{n=0}^{+\infty }\frac{\left( c\right) _{n}r^{n}}{n!}\;
{_2F_1}\left(-n,a,c,\xi \right)\; {_2F_1}\left(-n,b,c,\zeta \right)
&=  \frac{\left( 1-r\right) ^{a+b-c}}{\left( 1-r+\xi r\right) ^{a}\left(
1-r+\zeta r\right) ^{b}} \label{4.8} \\
&\times {_2F_1}\left(a,b,c;\frac{r\xi \zeta }{\left(
1-r+\xi r\right) \left( 1-r+\zeta r\right) }\right)  \nonumber
\end{align}
for $c=\gamma +1$, $r=e^{-\varepsilon }$, $a=b=\frac{\gamma }{2}+1$, $\xi
=1-e^{i\theta }$ and $\zeta =\overline{\xi },$ we arrive at the expression
\begin{equation}
\mathcal{N}_{\gamma ,\varepsilon }\left( \theta \right) =\frac{\left(
1-e^{-\varepsilon }\right) }{\left| 1-e^{-\varepsilon +i\theta }\right|
^{2+\gamma }}\; \; {_2F_1}\left(\frac{\gamma }{2}+1,\frac{\gamma }{2}+1,\gamma +1;%
\frac{e^{-\varepsilon }\left| 1-e^{i\theta }\right| ^{2}}{\left|
1-e^{-\varepsilon +i\theta }\right| ^{2}}\right)  \label{4.9}
\end{equation}
This ends the proof.
\end{proof}

\quad 

As mentioned in the introduction, if we consider the case $\gamma
=0,$ then one can check that the quantity in \eqref{4.9} equals to $%
\mathcal{N}_{0,\varepsilon }\left( \theta \right) =\left( 1-e^{-\varepsilon
}\right) ^{-1}$and the sequence of numbers in \eqref{4.2} become $%
\sigma _{0,\varepsilon }\left( n\right) $ =$e^{n\varepsilon }$ while the
circular Jacobi polynomials in \eqref{1.5} reduces to
\begin{equation}
g_{n}^{0}\left( e^{i\theta }\right) =\frac{\left( 1\right) _{n}}{n!}{}%
_{2}F_{1}(-n,1,1,1-e^{i\theta })=e^{in\theta }.  \label{4.10}
\end{equation}
Therefore, setting $\rho =e^{-\frac{1}{2}\varepsilon }<1,$ the constructed
states take the form:
\begin{equation}
\mid e^{i\theta },\varepsilon \left( \rho \right) ,0,\alpha >=\sqrt{1-\rho
^{2}}\sum\limits_{n=0}^{+\infty }\rho ^{n}e^{in\theta }\mid n,\alpha >.
\label{4.11}
\end{equation}
The latter can be considered as a phase coherent state expressed in the
basis of eigenstates of the Hamiltonian with PHO potential, whose form is
identical to the PCS in \eqref{1.3}.

\begin{proposition}
The states $\mid e^{i\theta },\varepsilon ,\gamma ,\alpha >$ satisfy the
following resolution of the identity
\begin{equation}
\lim_{\varepsilon \rightarrow 0^{+}}\int\limits_{0}^{2\pi }\mid e^{i\theta
},\varepsilon ,\gamma ,\alpha ><\alpha ,\gamma ,\varepsilon ,e^{i\theta
}\mid d\mu _{\gamma ,\varepsilon }\left( \theta \right) =\mathbf{1}%
_{L^{2}\left( \mathbb{R}_{+},dx\right) }  \label{4.12}
\end{equation}
where $\mathbf{1}_{L^{2}\left( \mathbb{R}_{+},dx\right) }$ is the identity
operator and $d\mu _{\gamma ,\varepsilon }\left( \theta \right) $ is a
measure on $\left[ 0,2\pi \right] $ with the expression
\begin{equation}
d\mu _{\gamma ,\varepsilon }\left( \theta \right) :=2^{\gamma }\frac{\Gamma
^{2}\left( \frac{1}{2}\gamma +1\right) }{\Gamma \left( \gamma +1\right) }%
\left( \sin \frac{\theta }{2}\right) ^{\gamma }\mathcal{N}_{\gamma
,\varepsilon }\left( \theta \right) \frac{d\theta }{2\pi },  \label{4.13}
\end{equation}
$\mathcal{N}_{\gamma ,\varepsilon }\left( \theta \right) $ being the
normalization factor given explicitly in \eqref{4.9}.
\end{proposition}

\begin{proof}
Let us assume that the measure takes the form
\begin{equation}
d\mu _{\gamma ,\varepsilon }\left( \theta \right) =\mathcal{N}_{\gamma
,\varepsilon }\left( \theta \right) \Omega _{\gamma }\left( \theta \right)
d\theta   \label{4.14}
\end{equation}
where $\Omega _{\gamma }\left( \theta \right) $ is an auxiliary density to
be determined. Let $\varphi \in L^{2}\left( \mathbb{R}_{+},dx\right) $ and
let us start by writing the following action
\begin{align}
\mathcal{O}_{\gamma ,\varepsilon }\left[ \varphi \right] &:=\left(
\int\limits_{0}^{2\pi }\mid e^{i\theta },\varepsilon ,\gamma ,\alpha
><\alpha ,\gamma ,\varepsilon ,e^{i\theta }\mid d\mu _{\gamma ,\varepsilon
}\left( \theta \right) \right) \left[ \varphi \right]   \label{4.15}
\\
&=\int\limits_{0}^{2\pi }<\varphi \mid e^{i\theta };\varepsilon ,\gamma
,\alpha ><\alpha ,\gamma ,\varepsilon ,e^{i\theta }\mid d\mu _{\gamma
,\varepsilon }\left( \theta \right)  . \label{4.16}
\end{align}
Making use Eq. \eqref{4.1}, we obtain successively
\begin{align}
\mathcal{O}_{\gamma ,\varepsilon }\left[ \varphi \right] &=\int\limits_{0}^{2%
\pi }<\varphi \mid \left( \mathcal{N}_{\gamma ,\varepsilon }\left( \theta
\right) \right) ^{-\frac{1}{2}}\sum\limits_{n=0}^{+\infty }\frac{\overline{%
g_{n}^{\gamma }\left( e^{i\theta }\right) }}{\sqrt{\sigma _{\gamma
,\varepsilon }\left( n\right) }}\mid n;\alpha >><\alpha ,\gamma ,\varepsilon
,e^{i\theta }\mid d\mu _{\gamma ,\varepsilon }\left( \theta \right)
\label{4.17}\\
&=\int\limits_{0}^{2\pi }\sum\limits_{n=0}^{+\infty }\frac{\overline{%
g_{n}^{\gamma }\left( e^{i\theta }\right) }}{\sqrt{\sigma _{\gamma
,\varepsilon }\left( n\right) }}<\varphi \mid n;\alpha ><\alpha ,\gamma
,\varepsilon ,e^{i\theta }\mid \left( \mathcal{N}_{\gamma ,\varepsilon
}\left( \theta \right) \right) ^{-\frac{1}{2}}d\mu _{\gamma ,\varepsilon
}\left( \theta \right)   \label{4.18}
\\
&=\left( \sum\limits_{m,n=0}^{+\infty }\int\limits_{0}^{2\pi }\frac{\overline{%
g_{n}^{\gamma }\left( e^{i\theta }\right) }g_{m}^{\gamma }\left( e^{i\theta
}\right) }{\sqrt{\sigma _{\gamma ,\varepsilon }\left( n\right) }\sqrt{\sigma
_{\gamma ,\varepsilon }\left( m\right) }}\mid n;\alpha ><\alpha ;m\mid
\left( \mathcal{N}_{\gamma ,\varepsilon }\left( \theta \right) \right)
^{-1}d\mu _{\gamma ,\varepsilon }\left( \theta \right) \right) \left[
\varphi \right] .  \label{4.19}
\end{align}
Replace $d\mu _{\gamma ,\varepsilon }\left( \theta \right) =\mathcal{N}%
_{\gamma ,\varepsilon }\left( \theta \right) \Omega _{\gamma }\left( \theta
\right) d\theta ,$ then Eq. \eqref{4.19} takes the form
\begin{equation}
\mathcal{O}_{\gamma ,\alpha ,\varepsilon }=\sum\limits_{m,n=0}^{+\infty }%
\left[ \int\limits_{0}^{2\pi }\frac{\overline{g_{n}^{\gamma }\left(
e^{i\theta }\right) }g_{m}^{\gamma }\left( e^{i\theta }\right) }{\sqrt{%
\sigma _{\gamma ,\varepsilon }\left( n\right) }\sqrt{\sigma _{\gamma
,\varepsilon }\left( m\right) }}\Omega _{\gamma }\left( \theta \right)
d\theta \right] \mid n;\alpha ><\alpha ;m\mid .  \label{4.20}
\end{equation}
Then, we need to consider the integral
\begin{equation}
I_{n,m}\left( \gamma ,\varepsilon \right) :=\frac{1}{\sqrt{\sigma _{\gamma
,\varepsilon }\left( n\right) }\sqrt{\sigma _{\gamma ,\varepsilon }\left(
m\right) }}\int\limits_{0}^{2\pi }\overline{g_{n}^{\gamma }\left( e^{i\theta
}\right) }g_{m}^{\gamma }\left( e^{i\theta }\right) \Omega _{\gamma }\left(
\theta \right) d\theta .  \label{4.21}
\end{equation}
We recall the orthogonality relations of circular Jacobi polynomials \cite[p.875]{10}):
\begin{equation}
2^{\gamma }\frac{\Gamma ^{2}\left( \frac{\gamma }{2}+1\right) }{2\pi }%
\int\limits_{0}^{2\pi }\overline{g_{n}^{\gamma }\left( e^{i\theta }\right) }%
g_{m}^{\gamma }\left( e^{i\theta }\right) \left( \sin \frac{\theta }{2}%
\right) ^{\gamma }dx=\frac{\Gamma \left( n+\gamma +1\right) }{n!}\delta
_{n,m}.  \label{4.22}
\end{equation}
This suggests us to set
\begin{equation}
\Omega _{\gamma }\left( \theta \right) :=2^{\gamma }\frac{\Gamma ^{2}\left(
\frac{\gamma }{2}+1\right) }{\Gamma \left( \gamma +1\right) }\left( \sin
\frac{\theta }{2}\right) ^{\gamma }\frac{d\theta }{2\pi }.  \label{4.23}
\end{equation}
Therefore, \eqref{4.21} reduces to
\begin{equation}
I_{n,m}\left( \gamma ,\varepsilon \right) =e^{-\varepsilon n}\frac{m!\Gamma
\left( n+\gamma +1\right) }{n!\Gamma \left( m+\gamma +1\right) }\delta _{n,m}
\label{4.24}
\end{equation}
which means that the operator in \eqref{4.15} takes the form:
\begin{align}
\mathcal{O}_{\gamma ,\alpha ,\varepsilon }\equiv \mathcal{O}_{\alpha
,\varepsilon }&=\sum\limits_{n,m=0}^{+\infty }e^{-n\varepsilon }\frac{%
m!\Gamma \left( n+\gamma +1\right) }{n!\Gamma \left( m+\gamma +1\right) }%
\delta _{n,m}\mid n;\alpha ><\alpha ;m\mid   \label{4.25}\\
&=\sum\limits_{m=0}^{+\infty }e^{-m\varepsilon }\mid m;\alpha ><\alpha ;m\mid
.  \label{4.26}
\end{align}
Thus, we arrive at \qquad\
\begin{equation}
\mathcal{O}_{\alpha ,\varepsilon }\left[ \varphi \right] =\sum%
\limits_{m=0}^{+\infty }e^{-m\varepsilon }\left( \mid m;\alpha ><\alpha
;m\mid \right) \left[ \varphi \right] .  \label{4.27}
\end{equation}
For $u\in \mathbb{R}_{+},$ we can write
\begin{align}
\mathcal{O}_{\alpha ,\varepsilon }\left[ \varphi \right] \left( u\right)
&=\sum\limits_{m=0}^{+\infty }e^{-m\varepsilon }<\varphi \mid m;\alpha
><u\mid m;\alpha >  \label{4.28}
\\ &
=\sum\limits_{m=0}^{+\infty }e^{-m\varepsilon }\left(
\int\limits_{0}^{+\infty }\varphi \left( v\right) \overline{<v\mid m;\alpha >%
}dv\right) <u\mid m;\alpha >  \label{4.29}
\\ &
=\int\limits_{0}^{+\infty }\varphi \left( v\right) \left(
\sum\limits_{m=0}^{+\infty }e^{-m\varepsilon }\overline{<v\mid m;\alpha >}%
<u\mid m;\alpha >\right) dv . \label{4.30}
\end{align}
We are then lead to calculate the sum
\begin{equation}
\mathcal{G}_{\varepsilon }^{\alpha }\left( u,v\right)
:=\sum\limits_{m=0}^{+\infty }e^{-m\varepsilon }\overline{<v\mid m;\alpha >}%
<u\mid m;\alpha >.  \label{4.31}
\end{equation}
For this we recall the expression of the eigenstate $\mid m;\alpha >$ in 
\eqref{2.8}. So that the above sum in \eqref{4.31} reads
\begin{eqnarray}
\mathcal{G}_{\varepsilon }^{\alpha }\left( u,v\right)  = 2\left( vu\right)
^{\alpha -\frac{1}{2}}e^{-\frac{1}{2}\left( u^{2}+v^{2}\right) }
 \sum\limits_{m=0}^{+\infty }e^{-m\varepsilon }\frac{m!}{\Gamma
\left( m+\alpha \right) }L_{m}^{\left( \alpha -1\right) }\left( u^{2}\right)
L_{m}^{\left( \alpha -1\right) }\left( v^{2}\right) .  \label{4.32}
\end{eqnarray}
Eq.\eqref{4.32} can be rewritten as
\begin{equation}
\mathcal{G}_{\varepsilon }^{\alpha }\left( u,v\right) =2\left( uv\right)
^{\alpha -\frac{1}{2}}e^{-\frac{1}{2}\left( u^{2}+v^{2}\right) }K\left(
e^{-\varepsilon };u^{2},v^{2}\right)   \label{4.33}
\end{equation}
where we have introduced the kernel function
\begin{equation}
K\left( \tau ;\xi ,\zeta \right) :=\sum\limits_{m=0}^{+\infty }\tau ^{m}%
\frac{m!}{\Gamma \left( m+\alpha \right) }L_{m}^{\left( \alpha -1\right)
}\left( \xi \right) L_{m}^{\left( \alpha -1\right) }\left( \zeta \right)
; \quad  0<\tau <1.  \label{4.34}
\end{equation}
The latter can be written in a closed form by applying the Hille-Hardy
formula \cite{17}. Now, returning back to Eq. \eqref{4.30}
and taking into account Eq.\eqref{4.33}, we get that

\begin{equation}
\mathcal{O}_{\alpha ,\varepsilon }\left[ \varphi \right] \left( u\right)
=2u^{\alpha -\frac{1}{2}}e^{-\frac{1}{2}u^{2}}\int\limits_{0}^{+\infty
}v^{\alpha -\frac{1}{2}}e^{-\frac{1}{2}v^{2}}K\left( e^{-\varepsilon
},u^{2},v^{2}\right) \varphi \left( v\right) dv . \label{4.35}
\end{equation}
Next,we split the right hand side of Eq.\eqref{4.35}
\begin{equation}
\mathcal{O}_{\alpha ,\varepsilon }\left[ \varphi \right] \left( u\right)
=\vartheta _{\alpha }\left( u\right) M\left[ \varphi \right] \left( u\right)
,  \label{4.36}
\end{equation}
where
\begin{equation}
M\left[ \varphi \right] \left( u\right) =\frac{1}{2}\int\limits_{0}^{+\infty
}K\left( \tau ;w,s\right) h\left( s\right) s^{\alpha -1}e^{-s}ds,  \label{4.37}
\end{equation}
with $\tau =e^{-\varepsilon }$, $w=u^{2}$ and
\begin{equation}
h\left( s\right) :=s^{-\frac{1}{2}\alpha +\frac{1}{4}}e^{\frac{1}{2}%
s}\varphi \left( \sqrt{s}\right) .  \label{4.38}
\end{equation}
By direct calculations, one can check that $h\in L^{2}\left( \mathbb{R}%
_{+},s^{\alpha -1}e^{-s}ds\right) .$ Precisely, we have that
\begin{equation}
\left\| h\right\| _{L^{2}\left( \mathbb{\,R}_{+},s^{\alpha
-1}e^{-s}ds\right) }^{2}=2\left\| \varphi \right\| _{L^{2}\left( \mathbb{R}%
_{+}\right) }  \label{4.39}
\end{equation}
We now apply the result of B. Muckenhoopt \cite{22} who considered
the Poisson integral of a function $f\in L^{p}\left( \mathbb{R}^{+},s^{\eta
}e^{-s}ds\right) ,\eta >-1,1\leq p\leq +\infty $ defined by
\begin{equation}
A\left[ f\right] \left( \tau ,w\right) :=\int\limits_{0}^{+\infty }K\left(
\tau ,w,s\right) f\left( s\right) s^{\eta }e^{-s}ds; \quad 0<\tau <1  \label{4.40}
\end{equation}
with the kernel $K\left( \tau ,\bullet ,\bullet \right) $ defined as in \eqref{4.34}. He proved that $\lim_{\tau \rightarrow 1^{-}}A\left[ f%
\right] \left( \tau ,y\right) =f\left( y\right) $ almost everywhere in $%
\left[ 0,+\infty \right[ ,1\leq p\leq \infty .$ We apply this result in the
case $p=2,f=h$ and $A\equiv M$ to obtain that
\begin{equation}
M\left[ \varphi \right] \left( u\right) \rightarrow 2^{-1}h\left(
u^{2}\right) =2^{-1}u^{-\alpha +\frac{1}{2}}e^{\frac{1}{2}u^{2}}\varphi
\left( u\right) .  \label{4.41}
\end{equation}
Recalling that $\tau =e^{-\varepsilon },$ we get that
\begin{equation}
\mathcal{O}_{\alpha ,\varepsilon }\left[ \varphi \right] \left( u\right)
=\vartheta _{\alpha }\left( u\right) M\left[ \varphi \right] \left( u\right)
\rightarrow \varphi \left( u\right) \quad \mbox{ as } \quad \varepsilon \rightarrow 0^{+}
\label{4.42}
\end{equation}
which means that
\begin{equation}
\lim_{\varepsilon \rightarrow 0^{+}}\int\limits_{0}^{2\pi }\mid e^{i\theta
},\varepsilon ,\gamma ,\alpha ><e^{i\theta },\varepsilon ,\gamma ,\alpha
\mid d\mu _{\gamma ,\varepsilon }\left( \theta \right) =\mathbf{1}%
_{L^{2}\left( \mathbb{R}_{+},dx\right) .}  \label{4.43}
\end{equation}
This ends the proof.
\end{proof}

\section{ A close form for the GPCS wavefunctions}

Now, we assume that the parameter $\gamma $ occurring in the definition of
the circular Jacobi polynomials $g_{n}^{\gamma }\left( e^{i\theta }\right) $
in \eqref{1.5} is connected to the parameter $\alpha $ controlling
the singular part $ax^{-2}$ of the Hamiltonian $\Delta _{a}$ in \eqref{1.4} by $\alpha =\gamma +1$ which means we are taking $\gamma =\frac{%
1}{2}\sqrt{1+4a}.$ Then, we can establish a closed form for the constructed
GPCS as follows.

\begin{proposition}
Let $\gamma =\frac{1}{2}\sqrt{1+4a}$ and $\varepsilon >0$ be fixed
parameters. Then, the wavefunctions of the states $\mid e^{i\theta
},\varepsilon ,\gamma >$ defined in \eqref{4.1} can be written in a
closed form as $<x\mid e^{i\theta },\varepsilon ,\gamma >=$\newline
\begin{equation}
\frac{\sqrt{2}\left( 1-e^{-\frac{1}{2}\varepsilon +i\theta }\right)
^{-1}x^{\gamma +\frac{1}{2}}\exp \left( -\frac{x^{2}}{2}\coth \frac{%
\varepsilon }{4}\right) ._{1}\digamma _{1}\left( 1+\frac{\gamma }{2}%
,1+\gamma ;\frac{\left( 1-e^{i\theta }\right) e^{-\frac{1}{2}\varepsilon
}x^{2}}{\left( 1-e^{-\frac{1}{2}\varepsilon }\right) \left( 1-e^{-\frac{1}{2}%
\varepsilon +i\theta }\right) }\right) }{\sqrt{\Gamma \left( \gamma
+1\right) }\left( \left( 1-e^{-\frac{1}{2}\varepsilon }\right) \left( 1-e^{-%
\frac{1}{2}\varepsilon +i\theta }\right) \right) ^{\frac{\gamma }{2}}\sqrt{%
\frac{\left( 1-e^{-\varepsilon }\right) {}{}}{\left| 1-e^{-\varepsilon
+i\theta }\right| ^{2+\gamma }}{}{}{}{}_{2}F_{1}(\frac{\gamma }{2}+1,\frac{%
\gamma }{2}+1,\gamma +1;\frac{e^{-\varepsilon }\left| 1-e^{i\theta }\right|
^{2}}{\left| 1-e^{-\varepsilon +i\theta }\right| ^{2}})}}  \label{5.1}
\end{equation}
for every $x\in \mathbb{R}_{+}$.
\end{proposition}

\begin{proof}
We start by writing the expression of the wave function of states $\mid
e^{i\theta },\varepsilon ,\gamma >$ according to Definition \eqref{4.1} as
\begin{equation}
<x\mid e^{i\theta },\varepsilon ,\gamma >=\left( \mathcal{N}_{\gamma
,\varepsilon }\left( \theta \right) \right) ^{-\frac{1}{2}%
}\sum\limits_{n=0}^{+\infty }\frac{g_{n}^{\gamma }\left( e^{i\theta }\right)
}{\sqrt{\sigma _{\gamma ,\varepsilon }\left( n\right) }}<x\mid n;\gamma
+1>; \quad x\in \mathbb{R}_{+}.  \label{5.2}
\end{equation}
We have thus to look for a closed form of the series
\begin{equation}
\mathcal{S}\left( x\right) :=\sum\limits_{n=0}^{+\infty }\frac{g_{n}^{\gamma
}\left( e^{i\theta }\right) }{\sqrt{\sigma _{\gamma ,\varepsilon }\left(
n\right) }}<x\mid n;\gamma +1>  \label{5.3}
\end{equation}
which also reads
\begin{equation}
\mathcal{S}\left( x\right) =\sum\limits_{n=0}^{+\infty }\frac{\left( \gamma
+1\right) _{n}}{n!\sqrt{\sigma _{\gamma ,\varepsilon }\left( n\right) }}{}%
_{2}F_{1}(-n,\frac{\gamma }{2}+1,\gamma +1,1-e^{i\theta })<x\mid n;\gamma
+1>.  \label{5.4}
\end{equation}
Replacing $\sigma _{\gamma ,\varepsilon }\left( n\right) $ and\ $<x\mid
n;\alpha >$ by their expressions in \eqref{4.1} and \eqref{2.8} respectively
\begin{equation}
\mathcal{S}\left( x\right) =\frac{\sqrt{2}x^{\gamma +\frac{1}{2}}e^{-\frac{1%
}{2}x}}{\sqrt{\Gamma \left( \gamma +1\right) }}\sum\limits_{n=0}^{+\infty
}e^{-\frac{1}{2}n\varepsilon }._{2}F_{1}(-n,\frac{\gamma }{2}+1,\gamma
+1,1-e^{i\theta })L_{n}^{\left( \gamma \right) }\left( x^{2}\right) .
\label{5.5}
\end{equation}
Put $\tau :=e^{-\frac{1}{2}\varepsilon },\left| \tau \right| <1.$ Then Equation
\eqref{5.5} becomes
\begin{equation}
\mathcal{S}\left( x\right) =\frac{\sqrt{2}x^{\gamma +\frac{1}{2}}e^{-\frac{1%
}{2}x^{2}}}{\sqrt{\Gamma \left( \gamma +1\right) }}\frak{S}\left( x\right),
\label{5.6}
\end{equation}
where
\begin{equation}
\frak{S}\left( x\right) :=\sum\limits_{n=0}^{+\infty }\tau ^{n}\text{ }%
_{2}\digamma _{1}(-n,\frac{\gamma }{2}+1,\gamma +1,1-e^{i\theta
})L_{n}^{\left( \gamma \right) }\left( x^{2}\right)  . \label{5.7}
\end{equation}
Next, we use the bilateral generating formula \cite[p.213]{23}:
\begin{equation}
\sum\limits_{n=0}^{+\infty }t^{n}\text{ }_{2}\digamma _{1}(-n,c,1+\nu
;y)L_{n}^{\left( \nu \right) }\left( u\right) =\left( 1-t\right) ^{-1+c-\nu
}\left( 1-t+yt\right) ^{-c}  \label{5.8}
\end{equation}
\begin{equation*}
\times \exp \left( \frac{-ut}{1-t}\right) \text{ }_{1}\digamma _{1}\left(
c,1+\nu ,\frac{yut}{\left( 1-t\right) \left( 1-t+yt\right) }\right)
\end{equation*}
for $t=\tau ,c=\frac{\gamma }{2}+1,y=1-e^{i\theta },\nu =\gamma $ and $u=x^{2},$ we obtain
\begin{eqnarray}
\frak{S}\left( x\right)  &=&\left( 1-\tau \right) ^{-\frac{1}{2}\gamma
}\left( 1-\tau e^{i\theta }\right) ^{-1-\frac{1}{2}\gamma }  \label{5.9} \\
&&\times \exp \left( \frac{-\tau x^{2}}{1-\tau }\right) ._{1}\digamma
_{1}\left( 1+\frac{\gamma }{2},1+\gamma ,\frac{\left( 1-e^{i\theta }\right)
\tau x^{2}}{\left( 1-\tau \right) \left( 1-\tau e^{i\theta }\right) }\right).
\notag
\end{eqnarray}
Summarizing up the above calculations as
\begin{equation}
<x\mid e^{i\theta },\varepsilon ,\gamma >=\left( \mathcal{N}_{\gamma
,\varepsilon }\left( \theta \right) \right) ^{-\frac{1}{2}}\frac{\sqrt{2}%
x^{\gamma +\frac{1}{2}}e^{-\frac{1}{2}x^{2}}}{\sqrt{\Gamma \left( \gamma
+1\right) }}\frak{S}\left( x\right) ,  \label{5.10}
\end{equation}
we arrive at
\begin{align}
<x\mid e^{i\theta },\varepsilon ,\gamma >&=\frac{\sqrt{2}x^{\gamma +\frac{1%
}{2}}e^{-\frac{1}{2}x^{2}}\left( 1-\tau \right) ^{-\frac{1}{2}\gamma }}{%
\sqrt{\Gamma \left( \gamma +1\right) }\left( \mathcal{N}_{\gamma
,\varepsilon }\left( \theta \right) \right) ^{\frac{1}{2}}\left( 1-\tau
e^{i\theta }\right) ^{1+\frac{1}{2}\gamma }}  \label{5.11} \\
&\times \exp \left( \frac{-\tau x^{2}}{1-\tau }\right) ._{1}\digamma
_{1}\left( 1+\frac{\gamma }{2},1+\gamma ,\frac{\left( 1-e^{i\theta }\right)
\tau x^{2}}{\left( 1-\tau \right) \left( 1-\tau e^{i\theta }\right) }\right)
\nonumber
\end{align}
Finally, we replace $\tau $ by $e^{-\frac{1}{2}\varepsilon }$ and the factor
$\mathcal{N}_{\gamma ,\varepsilon }\left( \theta \right) $ by its expression
in \eqref{4.3} to obtain \eqref{5.1}.
\end{proof}

Naturally, once we have obtained a closed \ form for the GPCS $\mid
e^{i\theta },\varepsilon ,\gamma >$ we can look for the associated coherent
state transform. In view of \eqref{4.1}, this transform should map
the space $L^{2}\left( \mathbb{R}_{+},dx\right) $ spanned by the eigenstates
$\mid n;\gamma +1>$ of the Hamiltonian in $\Delta _{a}$ onto the space in
which the coefficients $g_{n}^{\gamma }\left( e^{i\theta }\right) $ are
orthogonal, that is the space $L^{2}\left( S^{1},d\sigma _{\gamma
}\right) $ with $d\sigma _{\gamma }:=\frac{1}{2\pi }\sin ^{\gamma }\frac{%
\theta }{2}d\theta $. It should also obey the general form: $\varphi \mapsto
\sqrt{\mathcal{N}_{\gamma ,\varepsilon }\left( \theta \right) }\left\langle
e^{i\theta },\varepsilon ,\gamma \mid \varphi \right\rangle .$ Recalling
that the resolution of the identity, which usually ensures the isometry
property of such map, was obtained at the limit $\varepsilon \rightarrow
0^{+}$ in \eqref{4.12}, then a convenient definition for such transform could be
as follows.

\begin{definition}
Let $\gamma =\frac{1}{2}\sqrt{1+4a}$ be a fixed parameter. The coherent
state transform associated with the GPCS in \eqref{4.1} is the map
\begin{equation}
\mathcal{W}_{\gamma }:L^{2}\left( \mathbb{R}_{+},dx\right) \rightarrow
L^{2}\left( S^{1},d\sigma _{\gamma }\right)   \label{5.12}
\end{equation}
defined by
\begin{equation*}
\mathcal{W}_{\gamma }\left[ \varphi \right] \left( e^{i\theta }\right)
:=\lim_{\varepsilon \rightarrow 0^{+}}\int\limits_{0}^{+\infty }\frac{\sqrt{2%
}}{\sqrt{\Gamma \left( \gamma +1\right) }}x^{\gamma +\frac{1}{2}}\left(
1-e^{-\frac{1}{2}\varepsilon }\right) ^{-\frac{1}{2}\gamma }\left( 1-e^{-%
\frac{1}{2}\varepsilon +i\theta }\right) ^{-1-\frac{1}{2}\gamma }
\end{equation*}

\begin{equation*}
\times \exp \left( -\frac{1}{2}x^{2}\coth \frac{\varepsilon }{4}\right)
._{1}\digamma _{1}\left( 1+\frac{\gamma }{2},1+\gamma ,\frac{\left(
1-e^{i\theta }\right) e^{-\frac{1}{2}\varepsilon }x^{2}}{\left( 1-e^{-\frac{1%
}{2}\varepsilon }\right) \left( 1-e^{-\frac{1}{2}\varepsilon +i\theta
}\right) }\right) \overline{\varphi \left( x\right) }dx.
\end{equation*}
\end{definition}

In fact, this definition provides us with a new way of looking at the
circular Jacobi polynomials. Indeed, we establish the following precise fact.

\begin{proposition}
Let $\gamma =\frac{1}{2}\sqrt{1+4a}$ be a fixed parameter. Then, the
transform $\mathcal{W}_{\gamma }$ defined in (5.12) satisfies:
\begin{equation}
\mathcal{W}_{\gamma }\left[ x\mapsto \left\langle x\mid n;\gamma 
+1\right\rangle \right] \left( e^{i\theta }\right) =\frac{\sqrt{n!}}{\sqrt{%
\left( \gamma +1\right) _{n}}}g_{n}^{\gamma }\left( e^{i\theta }\right) ,
\label{5.13}
\end{equation}
for all $e^{i\theta }\in S^{1}.$ In other words, the normalized circular
Jacobi polynomials are the images of eigenstates of the Hamiltonian $\Delta
_{a}$ under the coherent states transform $\mathcal{W}_{\gamma }.$
\end{proposition}

\begin{proof}
\bigskip Let us write the  transform $\mathcal{W}_{\gamma}$ as
\begin{equation}
\mathcal{W}_{\gamma }\left[ x\mapsto \left\langle x\mid n;\gamma 
+1\right\rangle \right] \left( e^{i\theta }\right) =\lim_{\varepsilon
\rightarrow 0^{+}}\sqrt{\mathcal{N}_{\gamma ,\varepsilon }\left( \theta
\right) }\left\langle e^{i\theta },\varepsilon ,\gamma \mid n;\gamma 
+1\right\rangle .  \label{5.14}
\end{equation}
Next, we make use of the expression of the eigenstates in \eqref{2.8} and we shall calculate the quantity
\begin{equation}
\mathcal{Q}^{\left( \varepsilon \right) }:=\sqrt{\mathcal{N}_{\gamma
,\varepsilon }\left( \theta \right) }\left\langle e^{i\theta },\varepsilon
,\gamma \mid n;\gamma  +1\right\rangle   \label{5.15}
\end{equation}
in the right hand side of Eq.\eqref{5.14} without passing to the
limit with respect to $\varepsilon .$ So, for instance, we set $\tau =e^{-%
\frac{1}{2}\varepsilon }$ and rewrite \eqref{5.15} as
\begin{align}
\mathcal{Q}^{\left( \varepsilon \right) } &=\int\limits_{0}^{+\infty }\frac{%
\sqrt{2}x^{\gamma +\frac{1}{2}}\left( 1-\tau \right) ^{-\frac{1}{2}\gamma
}\left( 1-\tau e^{i\theta }\right) ^{-1-\frac{1}{2}\gamma }}{\sqrt{\Gamma
\left( \gamma +1\right) }}  \label{5.16} \\
&\times \exp \left( -\frac{1}{2}\left( \frac{1+\tau }{1-\tau }\right)
x^{2}\right) ._{1}\digamma _{1}\left( 1+\frac{\gamma }{2},1+\gamma ,\frac{%
\left( 1-e^{i\theta }\right) \tau x^{2}}{\left( 1-\tau \right) \left( 1-\tau
e^{i\theta }\right) }\right)  \nonumber \\
&\times \left( \frac{2n!}{\Gamma \left( \gamma +1+n\right) }\right) ^{\frac{%
1}{2}}x^{\gamma +\frac{1}{2}}e^{-\frac{1}{2}x^{2}}L_{n}^{\left( \gamma
\right) }\left( x^{2}\right) dx.  \nonumber
\end{align}
We set
\begin{equation}
\kappa :=\frac{\left( 1-e^{i\theta }\right) \tau }{\left( 1-\tau \right)
\left( 1-\tau e^{i\theta }\right) }  \label{5.17}
\end{equation}
and we rewrite \eqref{5.16} as follows
\begin{equation}
\mathcal{Q}^{\left( \varepsilon \right) }=\frac{2\sqrt{n!}\left( 1-\tau
\right) ^{-\frac{1}{2}\gamma }\left( 1-\tau e^{i\theta }\right) ^{-1-\frac{1%
}{2}\gamma }}{\sqrt{\Gamma \left( \gamma +1\right) }\sqrt{\Gamma \left(
\gamma +1+n\right) }}\frak{G}^{\left( \varepsilon \right) },  \label{5.18}
\end{equation}
where
\begin{eqnarray}
\frak{G}^{\left( \varepsilon \right) } &:&=\int\limits_{0}^{+\infty
}x^{2\gamma +1}\exp \left( -\frac{1}{1-\tau }x^{2}\right) L_{n}^{\left(
\gamma \right) }\left( x^{2}\right) \; {_{1}\digamma _{1}}\left( 1+\frac{\gamma }{2},1+\gamma ,\kappa
x^{2}\right) dx.  \label{5.19}
\end{eqnarray}
Writing the Laguerre polynomial in term of the $_{1}\digamma _{1}$ as in 
\eqref{2.5} and making the change of variable $t=x^{2},$ then \eqref{5.19} takes the form
\begin{equation}
\frak{G}^{\left( \varepsilon \right) }=\frac{\left( \gamma +1\right) _{n}}{%
2n!}\int\limits_{0}^{+\infty }t^{\gamma }e^{-\frac{1}{1-\tau }%
t}{}_{1}\digamma _{1}\left( -n,1+\gamma ,t\right) _{1}\digamma _{1}\left( 1+%
\frac{\gamma }{2},1+\gamma ,\kappa t\right) dt.  \label{5.20}
\end{equation}
Now, with the help of the formula (\cite[p.823]{24}):
\begin{eqnarray}
&&\int\limits_{0}^{+\infty }t^{c-1}e^{-st}{}_{1}\digamma _{1}\left( \beta
,c,t\right) _{1}\digamma _{1}\left( b,c,\lambda t\right) dt  \label{5.21} \\
&=&\Gamma \left( c\right) \left( s-1\right) ^{-\beta }\left( s-\lambda
\right) ^{-b}s^{\beta +b-c}._{2}\digamma _{1}\left( \beta ,b,c;\frac{\lambda
}{\left( s-1\right) \left( s-\lambda \right) }\right) ,  \notag
\end{eqnarray}
$\Re c>0$ and $\Re s>\Re \lambda +1,$ for the parameters: $%
c=\gamma +1,$ $\beta =-n,$ $b=1+\frac{\gamma }{2},$ $s=\frac{1}{1-\tau }$
and $\lambda =\kappa ,$ Eq.\eqref{5.20} becomes
\begin{equation}
\frak{G}^{\left( \varepsilon \right) }=\frac{\left( \gamma +1\right)
_{n}\Gamma \left( \gamma +1\right) \tau ^{n}\left( 1-\tau \right) ^{\gamma
+1}}{2n!\left( 1-\kappa \left( 1-\tau \right) \right) ^{1+\frac{\gamma }{2}}}%
._{2}\digamma _{1}\left( -n,1+\frac{\gamma }{2},1+\gamma ;\frac{\kappa
\left( 1-\tau \right) ^{2}}{\tau \left( 1-\kappa \left( 1-\tau \right)
\right) }\right) .  \label{5.22}
\end{equation}
Recalling the expression of $\kappa $ in \eqref{5.17} we find that
\begin{equation}
\frac{\kappa \left( 1-\tau \right) ^{2}}{\tau \left( 1-\kappa \left( 1-\tau
\right) \right) }=1-e^{i\theta }.  \label{5.23}
\end{equation}
We return back to \eqref{5.18} and we obtain that
\begin{eqnarray}
\mathcal{Q}^{\left( \varepsilon \right) } &=&\tau ^{n}\left( \frac{\left(
\gamma +1\right) _{n}}{n!}\right) ^{\frac{1}{2}}\frac{\left( 1-\tau \right)
^{\frac{\gamma }{2}+1}}{(\left( 1-\tau e^{i\theta }\right) \left( 1-\kappa
\left( 1-\tau \right) \right) )^{1+\frac{\gamma }{2}}}  \label{5.24} \\
&&\times _{2}\digamma _{1}\left( -n,1+\frac{\gamma }{2},1+\gamma
;1-e^{i\theta }\right)   \notag
\end{eqnarray}
Using again \eqref{5.17} we find by calculation
\begin{equation}
\frac{\left( 1-\tau \right) ^{\frac{\gamma }{2}+1}}{(\left( 1-\tau
e^{i\theta }\right) \left( 1-\kappa \left( 1-\tau \right) \right) )^{1+\frac{%
\gamma }{2}}}=1  \label{5.25}
\end{equation}
Recalling that $\tau =e^{-\frac{1}{2}\varepsilon }$, we arrive at
\begin{equation}
\mathcal{Q}^{\left( \varepsilon \right) }=e^{-\frac{1}{2}n\varepsilon
}\left( \frac{\left( \gamma +1\right) _{n}}{n!}\right) ^{\frac{1}{2}%
}._{2}\digamma _{1}\left( -n,1+\frac{\gamma }{2},1+\gamma ;1-e^{i\theta
}\right) .  \label{5.26}
\end{equation}
If let $\varepsilon \rightarrow 0^{+},$ then we obtain that
\begin{equation}
\mathcal{W}_{\gamma }\left[ x\mapsto \left\langle x\mid n;\alpha
+1\right\rangle \right] \left( e^{i\theta }\right) =\left( \frac{\left(
\gamma +1\right) _{n}}{n!}\right) ^{-\frac{1}{2}}g_{n}^{\gamma }\left(
e^{i\theta }\right) ,  \label{5.27}
\end{equation}
where the in the left hand side we have exactly the normalized circular
Jacobi polynomials.
\end{proof}


\begin{thebibliography}{99}

\bibitem{1} F. London,\textit{\ Z. Phys}. \textbf{40}, p.193, (1927)

\bibitem{2} A Wunshe, J.Opt.B: Quantum Semiclass. Opt. \textbf{3}
(2001) 206-218

\bibitem{3} A.V. Chizhov and M.G.A. Paris, acta phys.slovaca ,\textbf{%
\ 48 }No.3, 343-348 (1998)

\bibitem{4} L. Susskind and J. Glogower, Physica. \textbf{1}, p.49
(1964)

\bibitem{5} G. M. D'Ariano, M. F. Sacchi and M.G.A. Paris, \textit{%
Phys. Rev. A. }\textbf{57} (1998)

\bibitem{6} S.M. Barnett, \textit{J. Mod. Opt. A, 45, p.2201 (1998)}

\bibitem{7} H-C Fu and R. Sasaki, Negative binomial states of
quantized Radiation fields, Preprint YIPT-96-54, arXiv: quant-ph/ 9610024v1.

\bibitem{8} I.I. Gol'dman, V.D. Krivchenkov, V.I. Kogan and V.M.
Gakitskii, Problems in Quantum Mechanics ( London: Infosearch) p.8 (1960)

\bibitem{9} Mourad E.H.Ismail, Classical and Quantum Orthogonal
Polynomials in one variable, Encyclopedia of Mathematics and its
applications, Cambridge university press (2005)

\bibitem{10} LI-Chien Shen, orthogonal polynomials on the unit circle
associated with the Laguerre polynomials, Proc.Amer. Math. Soc.129, No. 3,
pp.873-879 (2000)

\bibitem{11} L.D.Abreu, Wavelet frames, Bergman spaces and Fourier
transforms of Laguerre functions, arXiv:0704.1487v1, math.CA 11 Apr 2007

\bibitem{12} N. S. Witte and P.J. Forrester, Gap probabilties in the
finite and scalled Cauchy randaom matrix ensembles, \textit{Nonlinearity},
13, pp.1965-1986

\bibitem{13}  Z Mouayn, \textit{J. Phys. A: Math. Gen. }\textbf{43 }%
(2010) 295201

\bibitem{14} I. I. Gol'dman and D. V.Krivchenkov, Problems in Quantum
Mechanics, Pergamon, London, 1961

\bibitem{15} D. Popov, \textit{. Phys. A: Math. Gen. }\textbf{34 }%
(2001) pp.1-14

\bibitem{16} R. L. Hall, N. Saad and A.B. Von Keviczky, \textit{J.
Phys. A: Math. Gen. }\textbf{34 }(2001), 11287-11300

\bibitem{17} W. Magnus, F.Oberhettinger \& R.P.Soni, Formulas and
Theorems for the Special Functions of Mathematical Physics, Springer-Verlag
Berlin Heidelberg New York, 1966.

\bibitem{18} K. Thirulogasantar and N. Saad, \textit{J. Phys. A:
Math. Gen. }\textbf{37 }(2004), 4567-4577

\bibitem{19} R.L. Hall, N. Saad \& A.B. von Kevicsky, Spiked harmonic
osicllators,\textit{\ arXiv: math-ph/ 0109014v1, 18 sep 2001}

\bibitem{20} V. V. Dodonov, 'Noncalssical' states in quantum optics:
a 'squeezed review of the first 75 years, \textit{J.opt.B:Quantum
Semiclass.opt. }\textbf{4}, R1-R33 (2002)

\bibitem{21} A. Erdely, W. Magnus, F. Oberhettinger and F.G. Tricomi,
Higher Transcendental Functions, Vol.1, Mc Graw-Hill, New York, 1953

\bibitem{22} B. Muckenhoupt, Poisson integrals for Hermite and
Laguerre expansions, Trans. Amer. Math. Soc, \textbf{139,}pp. 231-242\textbf{%
,} 1969

\bibitem{23} E. D. Rainville, Special functions, The Macmillan
company, New York, 1963 page 213

\bibitem{24} Gradshteyn I S and Ryzhik I M, ''Table of Integrals,
Series and Products'', Academic Press, INC, Seven Edition 2007
\end{thebibliography}
\end{document}